\documentclass[a4paper, 10pt, conference]{ieeeconf}      

\IEEEoverridecommandlockouts                              
\usepackage{amsfonts}
\usepackage{amsmath} 
\usepackage{graphicx}
\usepackage{ mathrsfs }
\usepackage{algorithm,algorithmic}

\renewcommand{\baselinestretch}{0.99}

\usepackage[usenames,dvipsnames,svgnames,table]{xcolor}
\usepackage{tikz}
\usetikzlibrary{automata,arrows,shapes,shapes.geometric}
\usetikzlibrary{fit}

\let\labelindent\relax
\usepackage{enumitem}

\newcommand{\ie}{\emph{i.e.}, }
\newcommand{\eg}{\emph{e.g.}, }

\newcommand{\nts}{\mathcal{N}}
\newcommand{\init}{\mathrm{init}}
\newcommand{\ap}{{AP}}

\newcommand{\obs}{O}
\newcommand{\sens}{\Theta}
\newcommand{\obsm}{M}
\newcommand{\obsf}{\gamma}
\newcommand{\cost}{g}
\newcommand{\strat}{C}

\DeclareMathOperator{\U}{\mathbf{U}}
\DeclareMathOperator{\X}{\mathbf{X}}
\DeclareMathOperator{\F}{\mathbf{F}}

\newcommand{\dfa}{\mathcal{A}}

\newcommand{\task}{\varphi}

\newcommand{\tac}{V}

\newcommand{\product}{\mathcal{P}}

\newcommand{\blf}{\mathcal{B}}
\newcommand{\B}{\mathbf{B}}
\newcommand{\A}{\mathbf{A}}
\newcommand{\bbf}{\mathbf{b}}
\newcommand{\abf}{\mathbf{a}}

\newcommand{\wtg}{\mathrm{wtg}}

\newtheorem{df}{Definition}
\newtheorem{proposition}{Proposition}
\newtheorem{corollary}{Corollary}
\newtheorem{theorem}{Theorem}
\newtheorem{problem}{Problem}
\newtheorem{remark}{Remark}
\newtheorem{example}{Example}

\newcommand{\eva}[1]{#1}

\title{\LARGE \bf
Optimal Observation \eva{Mode} Scheduling for Systems\\ under Temporal Constraints
}

\author{Eva Tesa\v{r}ov\'{a}, M\'{a}ria Svore\v{n}ov\'{a}, Ji\v{r}\'{i} Barnat, Ivana \v{C}ern\'{a}
\thanks{The authors are with Faculty of Informatics, Masaryk University, Brno, Czech Republic, {\tt\footnotesize xtesaro2@fi.muni.cz, svorenova@mail.muni.cz, barnat@fi.muni.cz, cerna@muni.cz}. This work has been partially supported by the Czech Science Foundation grant No. 15-08772S.}%
}

\begin{document}

\maketitle
\thispagestyle{empty}
\pagestyle{empty}

\begin{abstract}
Autonomous control systems use various sensors to decrease the amount of uncertainty under which they operate. While providing partial observation of the current state of the system, sensors require resources such as energy, time and communication. We consider discrete systems with non-deterministic transitions and multiple observation modes. The observation modes provide different information about the states of the system and are associated with non-negative costs. We consider two control problems. First, we aim to construct a control and observation mode switching strategy that guarantees satisfaction of a finite-time temporal property given as a formula of syntactically co-safe fragment of LTL (scLTL) and at the same time, minimizes the worst-case cost accumulated until the point of satisfaction. Second, the bounded version of the problem is considered, where the temporal property must be satisfied within given finite time bound. We present correct and optimal solutions to both problems and demonstrate their usability on a case study motivated by robotic applications.
\end{abstract}

\section{Introduction}

Embedded systems used in transportation, medical and other safety critical applications typically operate under uncertainty. The source of the uncertainty can be of two types. The internal uncertainty is bounded to the system's control inputs such as noisy actuators in mobile robots. The external uncertainty arises from the system's interaction with the environment such as other robots or people operating in the same space. In order to lower the amount of uncertainty, sensors are deployed to provide information about the current state of the system. Individual sensors and their combinations may provide varying, partial observation of the current state of the system. At the same time, their deployment requires resources such as energy, time or communication.

The field of sensor scheduling studies the problem of the  deployment of sensors in order to optimize estimation of a signal connected to the system's state. There is a wide range of results, \eg for linear systems~\cite{linearsensors}, hybrid systems~\cite{hybridsensors} and for applications in robot motion planning~\cite{mrg}. The related field of information gathering assumes a fixed set of sensors and aims to find a control strategy for the system that optimizes estimation of the signal. Recently, a problem combining optimization with temporal objectives has been considered in information gathering for discrete systems~\cite{austin}. 

Partial observability has been extensively studied for discrete systems in artificial intelligence and game theory. 
The main focus is typically on partially observable Markov decision processes (POMDPs) that model both partial observation and probabilistic uncertainties. The optimization objectives include expected total cost over finite horizon~\cite{martinaaai15}, and expected average or discounted total cost over infinite horizon~\cite{kaelbling98,pineau03}. Besides optimization objectives, many systems also operate under temporal constraints. Recently, temporal logics such as Linear Temporal Logic (LTL) or Computation Tree Logic (CTL) have been increasingly used to specify temporal properties of systems such as reachability, safety, stability of response. 
The overview of results for partially observable stochastic games (with POMDPs as a subclass) with respect to various classes of temporal objectives can be found in~\cite{pogamessurvey}. It is important to note that most of the problems of quantitative nature formulated for POMDPs are undecidable to solve precisely or even to approximate~\cite{pogamessurvey,madani03}. All the above results consider systems with one fixed observation mode that can be seen as a deployment of a single sensor.

Comparing to the aforementioned fields of study, in this work we focus on a problem that combines the optimal and temporal control for systems with multiple observation modes. We present a discrete system for modeling the above setting referred to as a \textit{non-deterministic transition system (NTS) with observation modes}. The non-determinism can be used to model both the internal and external uncertainty of the system whereas observation modes capture the sensing capabilities. In every step of an execution of the system, one decides which mode of partial observation to activate. Activation of each observation mode is associated with a non-negative cost. \eva{An example of a robotic system with limited energy resources and multiple sensing capabilities is a planetary rover. In~\cite{mrg}, the authors design an optimal schedule for the use of a localization system in a rover that minimizes energy consumption while at the same time guarantees safe path following. While sensor readings are typically continuous, in this work we assume that the set of readings that affect decision-making can be represented by a finite set, \eg sets of values satisfying the same constraints.}

We consider the following two problems. First, the aim is to construct a control and observation mode switching strategy for an NTS with observation modes that (i) guarantees satisfaction of a finite-time temporal property given as a formula of syntactically co-safe fragment of LTL (scLTL) and (ii) minimizes the worst-case cost accumulated until the point of satisfaction. The second problem considers the bounded version of the above problem, where the temporal property is required to be satisfied in at most $k\geq 1$ steps. Leveraging techniques from automata-based model checking and graph theory, we present correct and optimal solutions to both problems. While in this work, we restrict ourselves to objectives over finite time horizon, the aim is to use these results as a basis for solving more intriguing problems in our future work, \eg involving infinite-time temporal properties and cost functions, and probabilistic models. At the same time, temporal properties over finite horizon offer lower computational and strategy complexity compared to the general class of temporal properties and cover many interesting properties typically considered, \eg in robotic applications~\cite{scltl,austin,alphanscltl}. 
 
To the best of our knowledge, discrete systems with multiple modes were first considered only recently in \cite{krishbudget,krishminattention}, where the authors focus on control with respect to properties in infinite time horizon. The most related work to ours is~\cite{bertrand11} that considers a variation of POMDPs, where at each step the user can either choose to use the partial information or pay a fixed cost and receive the full information about the current state of the system. The authors discuss the problems of minimizing the worst-case or expected total cost before reaching a designated goal state with probability 1. While the former problem has optimal, polynomial solution, the latter proves undecidable. \eva{The main contribution of our work is twofold. First, we introduce a new model that extends the one in~\cite{bertrand11} in the sense that we allow multiple observation modes with varying costs. Second, we design correct and optimal strategies to control such models to guarantee an scLTL formula while minimizing the corresponding cost, over bounded or unbounded time horizon.  } 

The rest of this paper is organized as follows. In Sec.~\ref{sec:prelims}, we introduce NTS with observation modes and necessary definitions from temporal logic and automata theory. The two problems of interest are stated in Sec.~\ref{sec:pf} and solved in Sec.~\ref{sec:solution} and~\ref{sec:solutionbounded}, respectively. For better readability, we use an illustrative example to demonstrate the presented framework. In Sec.~\ref{sec:cs}, we evaluate the proposed algorithms on a case study motivated by robotic applications. We conclude with final remarks and future directions in Sec.~\ref{sec:conclusion}.


\section{Preliminaries}\label{sec:prelims}

\subsection{Notation}

For a set $X$, we use $X^*$ 
to denote the set of all finite 
sequences of elements of $X$. 
A finite sequence $\sigma=x_0\ldots x_n\in X^*$ has length $|\sigma|=n+1$, $\sigma(i)=x_i$ is the $i$-th element and $\sigma^i=\sigma(i)\ldots \sigma(n)$ is the suffix starting with the $i$-th element, for $0\leq i\leq n$. Similarly, for an infinite sequence $\rho=x_0x_1\ldots \in X^\omega$, $\rho(i)=x_i$ for all $i\geq 0$. A prefix of a finite sequence $\sigma$ or an infinite sequence $\rho$ is any sequence $\sigma(0)\ldots \sigma(k)$ for $0\leq k\leq |\sigma|$ or $\rho(0)\ldots \rho(k)$ for $k\geq 0$, respectively. 



\subsection{System with observation modes}

\begin{df}[NTS]\label{def:nts}
A non-deterministic transition system (NTS) is a tuple $\nts=(S,A,T,s_{\init},\ap,L)$, where $S$ is a non-empty finite set of states, $A$ is a non-empty finite set of actions, $T\colon S\times A\to 2^S$ 
is a transition function, $s_{\init} \in S$ is the initial state, $\ap$ is a set of atomic propositions, and $L\colon S \to 2^{\ap}$ is a labeling function.
\end{df}

A run of a NTS is an infinite sequence $s_0s_1\ldots \in S^\omega$ such that for every $i\geq 0$ there exists $a\in A$ with $s_{i+1}\in T(s_{i},a)$. A finite run is a finite prefix of a run of the NTS. 

\begin{df}[NTS with observation modes]\label{def:sens}
A NTS with observation modes is \eva{a tuple $(\nts, \obs, \obsm)$, where $\nts=(S,A,T,s_{\init},\ap,L)$ is NTS, $\obs$ is a non-empty finite set of observations  and $\obsm$ is a non-empty finite set  of observation modes}. Every observation mode $m\in \obsm$ is associated with an observation function $\obsf_m: S\to 2^\obs$ and a cost $\cost_m\in \mathbb{R}^+_{0}$.
\end{df}





A run of a NTS with observation modes is an infinite sequence $\rho = (s_0,m_0)(s_1,m_1)\ldots \in (S\times \obsm)^\omega$ such that $s_0s_1\ldots$ is a run of the NTS. A finite run $\sigma=(s_0,m_0)\ldots (s_n,m_n)\in (S\times \obsm)^*$ of the NTS with observation modes is a finite prefix of a run. A pair $(s,m)\in S\times \obsm$ of a state and an observation mode is called a configuration. 


Given a finite run $\sigma = (s_0,m_0)\ldots (s_n,m_n)$, we define the cost of $\sigma$ as follows
\begin{equation}\label{eq:toc}
\cost(\sigma)=\sum \limits_{i=0}^{n} \cost_{m_i}.
\end{equation}

The observational trace of a run $\rho = (s_0,m_0)(s_1,m_1)\ldots$ is the sequence $\obsf(\rho)=\obsf_{m_0}(s_0)\obsf_{m_1}(s_1)\ldots \in (2^\obs)^\omega$ and the propositional trace of $\rho$ is the sequence $L(\rho)=L(s_0)L(s_1)\ldots \in (2^\ap)^\omega$. The observational and propositional traces of finite runs are defined analogously.

\begin{df}[Strategy]\label{def:strat}
Given a NTS with observation modes, a (observation-based control and observation scheduling) strategy is a function $\strat:(2^\obs)^*\to A\times \obsm$ that defines the action and the observation mode to be applied in the next step based only on the sequence of past observations.
\end{df}

We use $\sigma_\strat$ and $\rho_\strat$ to denote finite and infinite runs of the NTS $\nts$ induced by a strategy $\strat$, respectively. Note that for every configuration $(s,m)$, the strategy $\strat$ induces a non-empty set of runs $\rho_\strat$ with $\rho_\strat(0)=(s,m)$.

\begin{example}\label{ex:system}
Consider an NTS $\nts=(S,A,T,s_{\init},\ap,L)$, where $S=\{s_1,\ldots ,s_7\}$, $A=\{a,b\}$, $s_\init=s_1$ and the transition function is as depicted in Fig.~\ref{fig:system}. We let $\ap=\{\star\}$ and the labeling function is indicated in Fig.~\ref{fig:system}, \ie $L(s_6)=\{\star\}$ and $L(s_i)=\emptyset$ for every $i\neq 6$. Consider three observation modes $\obsm=\{m_1,m_2,m_3\}$ for $\nts$ such that their respective observation functions $\obsf_1,\obsf_2,\obsf_3$ report neither the shape nor the color, only the shape and both the shape and the color of the state as shown in Fig.~\ref{fig:system}. Hence, the set of observations is 
\begin{align*}
\obs = \{ & \texttt{white}, \texttt{blue}, \texttt{red}, \\
& \texttt{circle}, \eva{\texttt{rectangle}}, \texttt{diamond}\}
\end{align*}
For example, for state $s_2$ the observation functions are defined as $\obsf_1(s_2)=\emptyset,\obsf_2(s_2)=\{\eva{\texttt{rectangle}}\}, \obsf_3(s_2)=\{\eva{\texttt{rectangle}},\texttt{blue}\}$. The costs of the observation modes are $\cost_1=0, \cost_2=1,\cost_3=2$.

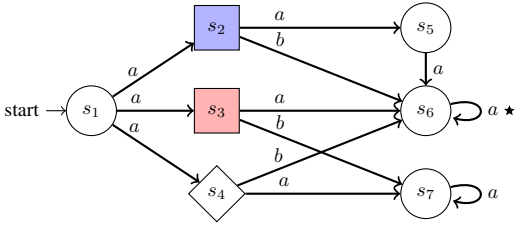
\begin{figure}[t]
\begin{center}
\tikzset{square blue state/.style={draw,regular polygon,regular polygon sides=4,fill=blue!30}}
\tikzset{square red state/.style={draw,regular polygon,regular polygon sides=4,fill=red!30}}
\tikzset{diamond state/.style={draw,diamond}}
\tikzset{star state/.style={draw,star,star points=5,star point ratio=2.25,fill=black,scale=.2}}
\begin{tikzpicture}[scale=0.55, every node/.style={scale=.75}]
\node [state,initial] (s1) at (-4,0) {$s_1$};
\node [square blue state] (s2) at (-1,2) {$s_2$};
\node [square red state] (s3) at (-1,0) {$s_3$};
\node [diamond state] (s4) at (-1,-2) {$s_4$};
\node [state] (s5) at (4,2) {$s_5$};
\node [state] (s6) at (4,0) {$s_6$};
\node [state] (s7) at (4,-2) {$s_7$};
\node [star state] (star) at (6,0) {};
\path[->, thick]
(s1) edge [above, near start]  node [align=center]  {$a$} (s2)
(s1) edge [above, near start]  node [align=center]  {$a$} (s3)
(s1) edge [above, near start]  node [align=center]  {$a$} (s4)
(s2) edge [above, near start]  node [align=center]  {$a$} (s5)
(s2) edge [above, near start]  node [align=center]  {$b$} (s6)
(s3) edge [above, near start]  node [align=center]  {$a$} (s6)
(s3) edge [above, near start]  node [align=center]  {$b$} (s7)
(s4) edge [above, near start]  node [align=center]  {$a$} (s7)
(s4) edge [above, near start]  node [align=center]  {$b$} (s6)
(s5) edge [right]  node [align=center]  {$a$} (s6)
(s6) edge [loop right]  node [align=center]  {$a$} (s6)
(s7) edge [loop right]  node [align=center]  {$a$} (s7)
;
\end{tikzpicture}
\end{center}
\caption{Example of an NTS with observation modes. For full description see Ex.~\ref{ex:system}.}
\label{fig:system}
\end{figure}
\end{example}

\subsection{Specification}\label{subsec:spec}

Linear temporal logic (LTL) is a modal logic with modalities referring to time~\cite{ltl}. Formulas of LTL are interpreted over infinite words such as the propositional traces generated by runs of a NTS with observation modes. Co-safe fragment of LTL, or co-safe LTL, contains all LTL formulas such that every satisfying infinite word has a good finite prefix~\cite{scltl}. A good finite prefix is a finite word such that every its extension to an infinite word satisfies the formula. A class of co-safe LTL formulas that are easy to characterize are syntactically co-safe LTL formulas~\cite{scltl}. 

\begin{df}[scLTL]\label{def:scltl}
Syntactically co-safe LTL (scLTL) formulas over $\ap$ are the LTL formulas formed as follows:
\begin{equation*}
\varphi :: p\mid \neg p\mid \varphi \wedge \varphi\mid \varphi \vee\varphi\mid \X\varphi \mid \varphi\U\varphi\mid \F\varphi,
\end{equation*}
where $p\in \ap$, $\wedge$ (conjunction) and $\vee$ (disjunction) are Boolean operators, and $\X$ (\emph{next}), $\U$ (\emph{until}) and $\F$ (\emph{future} or \emph{eventually}) are temporal operators.
\end{df}

The satisfaction relation $\models$ is recursively defined as follows. For a word $w\in (2^{\ap})^\omega$, we let:

\smallskip

\begin{tabular}{l c l}
$w\models p$ & $\Leftrightarrow$ & $p\in w(0)$,\\
$w\models \neg p$ & $\Leftrightarrow$ & $p\not \in w(0)$,\\
$w\models \varphi_1 \wedge \varphi_2$ & $\Leftrightarrow$ & $w\models\varphi_1$ and $w\models\varphi_2$,\\
$w\models \varphi_1 \vee \varphi_2$ & $\Leftrightarrow$ & $w\models\varphi_1$ or $w\models\varphi_2$,\\
$w\models \X \varphi$ & $\Leftrightarrow$ & $w^1\models \varphi$,\\
$w\models \varphi_1 \U\varphi_2$ & $\Leftrightarrow$ & there exists $i\geq 0 :\, w^i\models \varphi_2$,\\
& & and for all $0\leq j<i:\, w^j\models \varphi_1$\\
$w\models \F \varphi$ & $\Leftrightarrow$ & there exists $i\geq 0 :\, w^i\models \varphi$.
\end{tabular}

\smallskip
\eva{\begin{remark}
To express properties over bounded time horizon, bounded temporal opeators $\U_{\leq k}, \F_{\leq k}$ are often used in the literature. Note that these can be encoded using the operators from Def.~\ref{def:scltl}.
\end{remark}}

Even though scLTL formulas have infinite-time semantics, their satisfaction is guaranteed in finite time through the concept of good finite prefixes as explained above. We represent scLTL formulas with finite automata. 

\begin{df}[DFA]\label{def:dfa}
A deterministic finite automaton (DFA) is a tuple $\dfa=(Q,2^\ap,\delta,q_0,F)$, where $Q$ is a non-empty finite set of states, $2^\ap$ is the alphabet, $\delta:Q\times 2^\ap\to Q$ is a transition function, $q_0\in Q$ is the initial state and $F\subseteq Q$ is a non-empty set of accepting states.
\end{df}

A run of a DFA is a finite sequence $q_0q_1\ldots q_n\in Q^*$ such that for every $i\geq 0$, there exists $X\in 2^\ap$ such that $q_{i+1}=\delta(q_i,X)$. Every finite word $w\in (2^\ap)^*$ induces a run of the DFA. A run is called accepting if its last state is an accepting state. A word $w$ is accepted by the DFA if it induces an accepting run. 

Given an scLTL formula $\varphi$, one can construct a minimal (in the number of states) DFA that accepts all and only good finite prefixes of $\varphi$ using a translation algorithm from~\cite{scltl} and automata theory techniques~\cite{sipser}. 

A run $\rho$ of a NTS with observation modes satisfies an scLTL formula $\task$ if $L(\rho)\models \task$ or, equivalently, if there exists a finite prefix $\rho^\task$ of $\rho$ such that $L(\rho^\task)$ is a good finite prefiex for the formula $\task$. We refer to prefixes $\rho^\task$ as the good finite prefixes of the run $\rho$ for the formula $\task$. We say that a strategy $\strat$ satisfies $\task$ starting from a configuration $(s,m)$ if $\rho_\strat \models \task$ for every run $\rho_\strat$ such that $\rho_\strat(0)=(s,m)$.

\begin{example}\label{ex:task}
Consider the set of atomic propositions $\ap=\{\star\}$. An example of an scLTL formula over $\ap$ is $\task = \F \star$. A corresponding minimal DFA $\dfa$ for $\task$ is shown in Fig.~\ref{fig:task}.

\begin{figure}[t]
\label{fig:task}
\begin{center}
\tikzset{star/.style={draw,star,star points=5,star point ratio=2.25,fill=black,scale=.4}}
\begin{tikzpicture}[scale=0.8, every node/.style={scale=.7}]
\node[state, initial] (q0) at (0,0) {$q_0$};
\node[state, accepting](q1) at (3,0) {$q_1$};
;
\path[->, thick] (q0) edge [loop above]  node [align=center]  {$\emptyset$} (q0)
(q0) edge [above]  node [align=center]  {$\{\Large{\star}\}$} (q1)
(q1) edge [loop above]  node [align=center]  {$\emptyset,\{\star\}$} (q1);
\end{tikzpicture}
\end{center}
\vspace*{0.3cm}
\caption{A minimal DFA for the scLTL formula from Ex.~\ref{ex:task}.}
\end{figure}
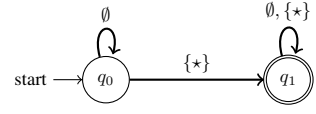
\end{example}


\section{Problem formulation}\label{sec:pf}

Consider a NTS with observation modes $(\nts, \obs, \obsm)$. The main motivation for the problem formulated in this section is a robotic system, \eg an autonomous car driving in an urban-like environment, that involves uncertainty originating, \eg from the motion of the robot such as noisy actuators of the car or from interaction with dynamic elements in the environment such as pedestrians on streets. Typically, the system is equipped with a set of sensors, where each sensor provides a partial information about the uncertainties. In such a case, an NTS can be used to model the motion capabilities of the robot in a partitioned environment and its interaction with the dynamic elements. The observation modes of the NTS then represent possible subsets of sensors and the cost of an observation mode corresponds to the amount of resources such as energy or communication needed to deploy the chosen set of sensors for a single step. During executions of the system, only the observations associated with the current state of the NTS and the chosen observation mode are available. Hence, the current state of the system might not be uniquely recognized.

We assume that the system is given a temporal objective in the form of an scLTL formula $\task$ over the set of atomic propositions $\ap$. Given a starting configuration $(s,m)$ and a strategy $\strat$ that satisfies the formula $\task$, we define the following cost function:
\begin{equation}\label{eq:tac}
\tac(\strat,(s,m),\task)=\max\limits_{\rho_\strat,\rho_\strat(0)=(s,m)}\quad \min \limits_{\rho^\task\rho'=\rho_\strat} \cost(\rho^\task).
\end{equation}
Intuitively, the cost $\tac(\strat,(s,m),\task)$ of a strategy $\strat$ with respect to the formula $\task$ and the configuration $(s,m)$ is the worst-case cumulative cost of the (earliest) satisfaction of $\task$ using $\strat$ starting from configuration $(s,m)$.

In this work, we aim to synthesize strategies that optimize the above cost, while at the same time guarantee satisfaction of the scLTL formula. We consider both general and bounded version of this problem, formulated as follows.

\smallskip

\begin{problem}[Optimal scLTL control]\label{pf:task}
Given 
a NTS \eva{with observation modes $(\nts, \obs, \obsm)$, where $\nts=(S,A,T,s_{\init},\ap,L)$,
an initial observation mode $m_\init \in \obsm$,} 
and an scLTL formula $\task$ over $\ap$, 
find an observation-based control and observation scheduling strategy $\strat$ such that 
(1) $\strat$ satisfies $\task$ starting from the configuration $(s_{\init},m_\init)$, 
and (2) the cost $\tac(\strat,(s_{\init},m_\init),\task)$ is minimized over all strategies satisfying $\task$ starting from the configuration $(s_{\init},m_\init)$.
\end{problem}

\smallskip

\begin{problem}[Bounded optimal scLTL control]\label{pf:taskbounded}
Given 
a NTS \eva{with observation modes $(\nts, \obs, \obsm)$, where $\nts=(S,A,T,s_{\init},\ap,L)$,
an initial observation mode $m_\init \in \obsm$,} 
an scLTL formula $\task$ over $\ap$, 
and a finite bound $k\geq 0$,  
find an observation-based control and observation scheduling strategy $\strat$ such that 
(1) $\strat$ satisfies $\task$ starting from the configuration $(s_{\init},m_\init)$ in at most $k$ steps, 
and (2) the cost $\tac(\strat,(s_{\init},m_\init),\task)$ is minimized over all strategies satisfying $\task$ starting from the configuration $(s_{\init},m_\init)$ in at most $k$ steps.
\end{problem}

\smallskip

In Sec.~\ref{sec:solution}, we propose an algorithm to solve the general Problem~\ref{pf:task} and prove its correctness and optimality. The algorithm builds on techniques from automata-based model checking and graph theory. The bounded Problem~\ref{pf:taskbounded} can be solved using an alternation of the above algorithm as proposed in Sec.~\ref{sec:solutionbounded}. Both algorithms are demonstrated in Sec.~\ref{sec:cs} on an illustrative case study of a mobile robot in an indoor environment equipped with a set of sensors. 

\begin{example}\label{ex:tac}
Consider the NTS with observation modes introduced in Ex.~\ref{ex:system} with initial observation mode $m_1$ and the scLTL formula from Ex.~\ref{ex:task} that requires to reach the state labeled with $\star\,$, \ie state $s_6$. Note that only in states $s_2,s_3,s_4$ there is more than one action allowed and hence it suffices to discuss strategies based on their decision in these states. No strategy $\strat$ with $\strat(\emptyset)=(a,m_1)$, \ie a strategy that applies action $a$ starting from the initial state $s_1$ and activates observation mode $m_1$ in the next state, can guarantee satisfaction of the formula. The reason is that the three states $s_2,s_3,s_4$ cannot be told apart using mode $m_1$ and both actions $a,b$ always in at least one case lead to state $s_7$ from which $s_6$ cannot be reached. Consider strategy $\strat_1$ that recognizes the shape of the three states $s_2,s_3,s_4$, \ie
\begin{align*}
\strat_1(\emptyset) &= (a,m_2),\\
\strat_1(\emptyset\{\eva{\texttt{rectangle}}\}) &= (a, m_1),\\
\strat_1(\emptyset\{\texttt{diamond}\}) &= (b, m_1).
\end{align*}
Strategy $\strat_1$ guarantees a visit to $s_6$ in at most 3 steps and its cost $\tac(\strat_1,(s_\init,m_1),\task)=1$. Alternatively, consider strategy $\strat_2$ that recognizes both the shape and the color of the three states $s_2,s_3,s_4$, \ie
\begin{align*}
\strat_1(\emptyset) &= (a,m_3),\\
\strat_1(\emptyset\{\eva{\texttt{rectangle}},\texttt{blue}\}) &= (b, m_1),\\
\strat_1(\emptyset\{\eva{\texttt{rectangle}},\texttt{red}\}) &= (a, m_1),\\
\strat_1(\emptyset\{\texttt{diamond,\texttt{white}}\}) &= (b, m_1).
\end{align*}
Strategy $\strat_2$ guarantees a visit to $s_6$ in 2 steps and its cost $\tac(\strat_2,(s_\init,m_1),\task)=2$. Strategy $\strat_1$ is the solution to the optimal scLTL control Problem~\ref{pf:task} as its cost is lower than the cost of $\strat_2$. However, if we consider the bounded optimal scLTL control Problem~\ref{pf:taskbounded} with $k=2$, then $\strat_2$ is the solution as $\strat_1$ may need more than 2 steps to reach $s_6$.  
\end{example}



\section{Optimal scLTL control}\label{sec:solution}



In this section, we describe the algorithm to solve Problem~\ref{pf:task} in detail. To approach the problem, we leverage automata-based model checking techniques that analyze the state space using graph algorithms. We first construct a synchronous product of the NTS $\nts$ and a DFA $\dfa$ for the scLTL formula $\task$, where the runs of the NTS satisfying the formula can be easily identified through accepting states of the DFA. Next, to account for the non-determinism and partial observation, we use a belief construction over the product that determines the set of possible current states of the product given any finite sequence of past observations. Using graph algorithms, we construct a strategy for the belief product that guarantees a visit of an accepting state and minimizes a function derived from the costs of the associated observation modes. Finally, we map the strategy from the belief product to the original NTS and prove that the resulting strategy is a solution to Problem~\ref{pf:task}. 

\subsection{Constructing the product}\label{subsec:product}

\begin{df}[Product]\label{def:product}
Let $\nts=(S,A,T,s_{\init},\ap,L)$ be a NTS and $\dfa = (Q, 2^{\ap}, \delta, q_0, F)$ be a DFA. The synchronous product is a tuple
\vspace{-0.2cm}
\begin{equation*}
\product = \nts \times \dfa = (S \times Q, A, T_{\product}, (s_{\init},q_0),\ap,L_\product, F_{\product}),
\end{equation*}
\vspace{-0.6cm}\\
where 
\begin{itemize}[noitemsep,topsep=0pt,parsep=0pt,partopsep=0pt]
\item $S\times Q$ is the set of states, 
\item $A$ is the alphabet, 
\item $T_{\product}\colon S \times Q\times A\to 2^{S\times Q}$ is a transition function such that $(s',q')\in T_{\product}((s,q),a)$ if and only if $s'\in T(s,a)$ and $\delta(q,L(s)) = q'$, 
\item $(s_{\init},q_0)$ is the initial state, 
\item $L_\product\colon S\times Q\to 2^{\ap}$ is the labeling function such that $L_\product((s,q))=L(s)$, 
\item $F_{\product} = \{(s,q) \mid q \in F\}$ is the set of product accepting states.
\end{itemize}
\end{df}

Note that the product can be seen as an NTS with a set of accepting states. This allows us to adopt the definitions of an infinite and finite runs for the product as well as a notion of an accepting finite run.

\begin{example}\label{ex:product}

In Fig..~\ref{fig:product}, we depict the product constructed for the NTS with observation modes presented in Ex.~\ref{ex:system} and the DFA from Ex.~\ref{ex:task}.

\begin{figure}[t]
\begin{center}
\tikzset{square blue state/.style={draw,rectangle,fill=blue!30,inner sep=10}}
\tikzset{square red state/.style={draw,rectangle,fill=red!30,inner sep=10}}
\tikzset{diamond state/.style={draw,diamond}}
\tikzset{star state/.style={draw,star,star points=5,star point ratio=2.25,fill=black,scale=.2}}
\begin{tikzpicture}[scale=0.7, every node/.style={scale=.55}]
\node [state,initial] (s1) at (-4,0) {\large $(s_1,q_0)$};
\node [square blue state] (s2) at (-2,2) {\large $(s_2,q_0)$};
\node [square red state] (s3) at (-2,0) {\large $(s_3,q_0)$};
\node [diamond state] (s4) at (-2,-2) {\large $(s_4,q_0)$};
\node [state] (s5) at (2,2) {\large $(s_5,q_0)$};
\node [state] (s6) at (2,0) {\large $(s_6,q_0)$};
\node [state] (s7) at (2,-2) {\large $(s_7,q_0)$};
\node [state, accepting] (s8) at (4,0) {\large $(s_7,q_1)$};
\node [star state] (star) at (6,0) {};
\path[->, thick]
(s1) edge [above, near start]  node [align=center]  {$a$} (s2)
(s1) edge [above, near start]  node [align=center]  {$a$} (s3)
(s1) edge [above, near start]  node [align=center]  {$a$} (s4)
(s2) edge [above, near start]  node [align=center]  {$a$} (s5)
(s2) edge [above, near start]  node [align=center]  {$b$} (s6)
(s3) edge [above, near start]  node [align=center]  {$a$} (s6)
(s3) edge [above, near start]  node [align=center]  {$b$} (s7)
(s4) edge [above, near start]  node [align=center]  {$a$} (s7)
(s4) edge [above, near start]  node [align=center]  {$b$} (s6)
(s5) edge [right]  node [align=center]  {$a$} (s6)
(s6) edge [above]  node [align=center]  {$a$} (s8)
(s7) edge [loop right]  node [align=center]  {$a$} (s7)
(s8) edge [loop right]  node [align=center]  {$a$} (s8)
;
\end{tikzpicture}
\end{center}
\caption{\eva{Product of the NTS from Ex.~\ref{ex:system} and the DFA from Ex.~\ref{ex:task}.}}
\label{fig:product}
\end{figure}
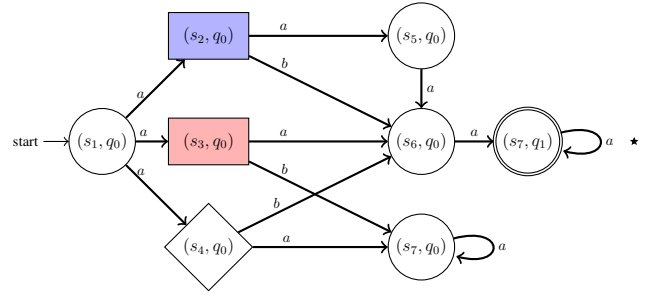

\end{example}

\smallskip

We abuse the notation by using $\obsf_\alpha$ to denote the observation function of a sensor $\alpha\in \sens$ as well as its extension to $S\times Q$, \ie $\obsf_\alpha((s,q))=\obsf_\alpha(s)$ for all $(s,q)\in S\times Q$.

\subsection{Constructing the belief product}\label{subsec:belief}

The belief construction over the product follows the standard principles used for partially observable systems. Besides keeping track of the states that the product can currently be in, we also keep track of the observation mode deployed in the current state.

\begin{df}[Weighted belief product]\label{def:belief}
Given a product $\product$ 
 and a set of observation modes $\obsm$ with the initial observation mode $m_\init$, we define the weighted belief product
\vspace{-0.3cm}
\begin{equation*}
\blf = (\B,\A, T_{\blf},\bbf_\init, \obs, F_{\blf}, w)
\end{equation*}
\vspace{-0.6cm}\\
over $\product$, where 
\begin{itemize}[noitemsep,topsep=0pt,parsep=0pt,partopsep=0pt]
\item $\B\subseteq 2^{S\times Q}$ is the set of all belief states, where a belief state $\bbf\in \B$ is a set of product states such that there exists an observation mode $m\in \obsm$ such that all states in $\bbf$ have the same observations in mode $m$, 
\item $\A=A\times \obsm$ is the set of belief actions of the form $\abf=(a,m)$, where $a\in A$ is an action of $\product$ and $m\in \obsm$ is an observation mode, 
\item $T_{\blf} : \B \times \A \to 2^\B$ is the transition function such that a belief state $\bbf'\in T_{\blf}(\bbf,(a,m'))$ if and only if $\bbf'$ is the set of all product states that can be reached in one step from a state in $\bbf$ using action $a$ and have the same observations in mode $m'$, 
\item $\bbf_\init=\{(s_\init,q_0)\}$ is the initial state,
\item $F_{\blf} = \{\bbf\mid \bbf\subseteq F_\product\}$ is the set of accepting belief states,
\item $w: \B \times \A \to \mathbb{R}^+_0$ is the weight function such that $w(\bbf,(a,m)) = \cost_{m}$.
\end{itemize}
\end{df}


Weighted belief product can be seen as a NTS with a set of accepting states and a weight function on transitions. We adopt definitions of finite and infinite runs of the belief product and an accepting finite run. 

\begin{corollary}\label{cor:belief}
From the definition of the weighted belief product $\blf$ it follows that every finite run of $\blf$ corresponds to exactly one finite sequence of observations in $(2^\obs)^*$ and at the same time, every finite sequence of observations in $(2^\obs)^*$ corresponds to at most one finite run of $\blf$.
\end{corollary}

\begin{example}\label{ex:belief}

In Fig..~\ref{fig:belief}, we depict part of the weighted belief product constructed for the NTS with observation modes presented in Ex.~\ref{ex:system} and the DFA from Ex.~\ref{ex:task}. 

\begin{figure}[t]
\begin{center}
\tikzset{state/.style={draw,ellipse, inner sep=10}}
\tikzset{square state/.style={draw,rectangle,inner sep=10}}
\tikzset{square blue state/.style={draw,rectangle,fill=blue!30,inner sep=10}}
\tikzset{square red state/.style={draw,rectangle,fill=red!30,inner sep=10}}
\tikzset{diamond state/.style={draw,diamond}}
\tikzset{star state/.style={draw,star,star points=5,star point ratio=2.25,fill=black,scale=.2}}
\begin{tikzpicture}[scale=0.7, every node/.style={scale=.55}]
\node [state,initial] (s1) at (-4,3) {\large $\{(s_1,q_0)\}$};
\node [state] (s2) at (3,3) {\large $\{(s_2,q_0),(s_3,q_0),(s_4,q_0)\}$};
\node [square state] (s3) at (3,1) {\large $\{(s_2,q_0),(s_3,q_0)\}$};
\node [diamond state] (s4) at (3,-1) {\large $\{(s_4,q_0)\}$};
\node [square blue state] (s5) at (-0.5,-1) {\large $\{(s_2,q_0)\}$};
\node [square red state] (s6) at (-4,-1) {\large $\{(s_3,q_0)\}$};
\path[->, thick]
(s1) edge [above]  node [align=center]  {$(a,m_0), \mathbf{0}$} (s2)
(s1) edge [sloped, above]  node [align=center]  {$(a,m_1), \mathbf{1}$} (s3)
(s1) edge [sloped, above]  node [align=center]  {$(a,m_1),\mathbf{1}$} (s4)
(s1) edge [sloped, below]  node [align=center]  {$(a,m_2),\mathbf{2}$} (s4)
(s1) edge [sloped, above]  node [align=center]  {$(a,m_2), \mathbf{2}$} (s5)
(s1) edge [sloped, above]  node [align=center]  {$(a,m_2), \mathbf{2}$} (s6)
;
\end{tikzpicture}
\end{center}
\caption{Part of the weighted belief product for the NTS with observation modes presented in Ex.~\ref{ex:system} and the DFA from Ex.~\ref{ex:task}. The costs of individual belief actions are written in bold.}
\label{fig:belief}
\end{figure}
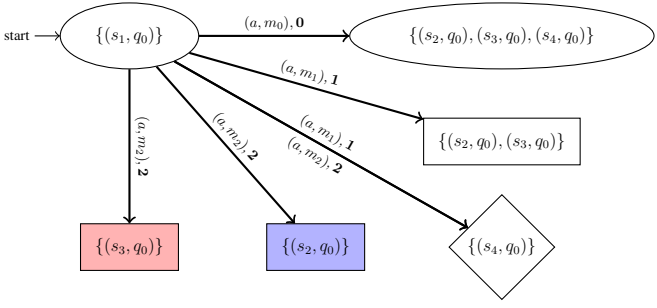

\end{example}

\smallskip

\begin{df}[Strategy]
Given a weighted belief product $\blf = (\B,\A, T_{\blf},\bbf_\init,\obs, F_{\blf}, w)$, a strategy for $\blf$ is a function $\strat\colon \eva{\B^* \to \A}$.\end{df}

\eva{We call a strategy $\strat$ memoryless if it can be defined as a function $\strat\colon \B \to \A$. }
If the context is clear, we use $\sigma_\strat$ and $\rho_\strat$ to denote finite and infinite runs of $\blf$ induced by a strategy $\strat$, respectively. 

\subsection{Constructing a strategy for the belief product}\label{subsec:beliefstrat}

In this section, we propose an algorithm that constructs a \eva{memoryless} strategy for the weighted belief product that guarantees a visit to an accepting state (if such a strategy exists) and minimizes the worst-case cumulative weight. We prove that such a strategy then maps to a strategy for the original NTS with observation modes that solves Problem~\ref{pf:task}. The algorithm can be seen as a combination of the standard algorithm for computing winning states in non-deterministic systems~\cite{baierbook} and Dijkstra's algorithm for computing shortest paths in a weighted graph~\cite{algorithms}. 

In the algorithm, we incrementally compute a value $\wtg(\bbf)$ (weight-to-go) for every belief state $\bbf$ that is the minimum worst case weight of reaching an accepting state starting from $\bbf$. Initially, the value is $0$ for accepting belief states and $\infty$ otherwise. We use $W_i$ to denote the set of belief states for which the value $\wtg(\bbf)\neq \infty$ after $i$-th iteration. In $i$-th iteration, we consider the belief state $\bbf_{\min}\in \B\setminus W_{i-1}$ and its action $\abf_{\bbf_{\min}}$ that leads to the set $W_{i-1}$ and minimizes the worst-case sum of the weight of the action and the value $\wtg$ of a successor state. The algorithm terminates when the initial belief state $\bbf_\init$ is added to the set $W_i$ or when there exists no state $\bbf\in \B\setminus W_i$ with an action leading to $W_{i}$. If the resulting set $W_i$ contains the initial belief state, the strategy consisting of the above actions for each belief state in $W_i$ is returned. The algorithm is summarized in Alg.~\ref{alg:1}. 

\begin{algorithm}[t]
\small
\begin{algorithmic}[1]
\REQUIRE $\blf = (\B,\A, T_{\blf},\bbf_\init,\obs, F_{\blf}, w)$
\ENSURE \eva{memoryless} strategy $\strat$ for the belief product $\blf$ 

\STATE $W_0 := F_{\blf}$
\STATE $\forall\, \bbf \in W_0:\, \wtg(\bbf) := 0$\\ $\forall\, \bbf \in (\B \backslash W_0):\, \wtg(\bbf) := \infty$
\STATE $i:=1$

\WHILE{$\bbf_\init\not \in W_i$ and exist $\bbf\in \B\backslash W_{i-1}$ and $\abf\in \A$ such that $\emptyset\neq T_\blf(\bbf,\abf) \subseteq W_{i-1}$}
\STATE $\bbf_{\min}:=\bot\qquad \abf_{\min}:=\bot\qquad \Delta_{\min}:=\infty$
\FOR{every $\bbf\in \B\backslash W_{i-1}$ and $\abf\in \A$ such that $\emptyset\neq T_\blf(\bbf,\abf) \subseteq W_{i-1}$}
\STATE $\Delta := \max \limits_{\bbf' \in T_\blf(\bbf,\abf)} \{w(\bbf,\abf) + \wtg(\bbf')\}$
\IF{$\Delta< \Delta_{\min}$}
\STATE $\bbf_{\min}:=\bbf\qquad \abf_{\min}:=\abf\qquad \Delta_{\min} := \Delta$
\ENDIF
\ENDFOR

\STATE $W_i := W_{i-1} \cup \{\bbf_{\min}\}$
\STATE $\strat(\bbf_{\min}) := \abf_{{\min}}$
\STATE $\wtg(\bbf_{\min}) := \Delta_{{\min}}$
\STATE $i:=i+1$
\ENDWHILE

\IF{$\bbf_\init \in W_i$}
\RETURN $\strat$
\ELSE
\RETURN no suitable strategy exists
\ENDIF
\end{algorithmic}
\caption{Constructing a strategy for the weighted belief product that maps to a solution of Problem~\ref{pf:task}.}\label{alg:1}
\end{algorithm}

\begin{proposition}[Correctness]\label{prop:corr1}
Alg.~\ref{alg:1} results in a strategy $\strat$ for the weighted belief product such that every run under $\strat$ that starts in $\bbf_\init$ eventually visits an accepting belief state, if such a strategy exists.
\end{proposition}

\begin{proof}
The property can be proved by induction on the iteration counter $i$ and proving that starting from the state $\bbf_i\in W_i, \bbf_i\not \in W_{i-1}$, strategy $\strat$ guarantees a visit to an accepting belief state in at most $i$ steps. 
\end{proof}

\begin{proposition}[Optimality]\label{prop:opt1}
Let $\strat$ be the strategy resulting from Alg.~\ref{alg:1}. Then among all strategies that guarantee a visit to an accepting belief state, $\strat$ minimizes the value 
\begin{equation}\label{eq:opt}
\tac_\blf(\strat,\bbf_\init) = \max \limits_{\rho_{\strat}, \rho_{\strat}(0)=\bbf_\init}\quad \min \limits_{\rho^{\mathrm{acc}}\rho'=\rho_\strat} w(\rho^\mathrm{acc},\strat)
\end{equation}
where $\rho^\mathrm{acc}$ is a finite run ending in an accepting state and
\begin{equation*}
w(\rho^\mathrm{acc},\strat)=\sum \limits_{i=0}^{|\rho^{\mathrm{acc}}|-2} w\big(\rho^\mathrm{acc}(i),\strat(\eva{\rho^\mathrm{acc}(0)\dots\rho^\mathrm{acc}(i)})\big).
\end{equation*}
Intuitively, the value $V_\blf(\strat,\bbf)$ of a strategy $\strat$ with respect to a belief state $\bbf$ is the worst-case cumulative weight of the (earliest) visit to an accepting state using $\strat$ starting from $\bbf$.
\end{proposition}

\begin{proof}
\eva{We show by induction that after every iteration $i\geq 1$, it holds that $\tac_\blf(\strat,\bbf_i)=\wtg(\bbf_i)$ for all states $\bbf_i\in W_i$, \ie the strategy $\strat$ realizes the values $\wtg(\bbf_i)$, and that the strategy $\strat$ minimizes the value $\tac_\blf(\cdot,\bbf_i)$ among all strategies that guarantee visit to an accepting state.}

\eva{
Assume that the strategy $\strat$ is computed in $n$ iterations of the ``while'' cycle in line 4, \ie $W_n$ is the resulting fixed point set. Consider a belief state $\bbf_i$ that was added to the set $W_n$ in $i$-th iteration, \ie $\bbf_i\not \in W_{i-1}, \bbf_i\in W_i$. Assume that for all $j< i$ it holds that $\strat$ minimizes the value $\tac_\blf(\cdot,\bbf_j)$ for every $\bbf_j \in W_j$ and that $\tac_\blf(\strat,\bbf_j)=\wtg(\bbf_j)$. Trivially, $\strat$ minimizes the value for all accepting belief states $\bbf\in F_\blf$ as $\tac_\blf(\strat,\bbf)=0=\wtg(\bbf)$. We show that $\strat$ then also minimizes the value $\tac_\blf(\cdot,\bbf_i)$ over all strategies and $\tac_\blf(\strat,\bbf_i) = \wtg(\bbf_i)$.}

\eva{
Assume by contradiction that there exists a (possibly not memoryless) strategy $\strat'$ for $\blf$  such that $\tac_\blf(\strat',\bbf_i)<\tac_\blf(\strat,\bbf_i)$. As $\strat$ is optimal for all $\bbf_j, j<i$, it must hold that  
\begin{equation}\label{eq:opthelp}
\tac_\blf(\strat',\bbf_j)=\wtg(\bbf_j)=\tac_\blf(\strat,\bbf_j).
\end{equation}
Let $\bbf \notin W_{i-1}$ be a belief state such that there exists a run $\sigma_{\strat'}$ under $\strat'$ that leads from $\bbf_i$ through $\bbf$ to an accepting belief state and $T_\blf(\bbf,\strat'(\sigma^\bbf_{\strat'}))\subseteq W_{i-1}$, where $\sigma^\bbf_{\strat'}$ is a prefix of $\sigma_{\strat'}$ ending in the state $\bbf$. Note that such $\bbf$ must exist since $\strat'$ guarantees a visit to an accepting state and $F_\blf\subseteq W_{i-1}$ (be aware that $\bbf$ can be $\bbf_i$ itself). Since the cumulative weight $w(\sigma^\bbf_{\strat'},\strat')$ is non-negative and the action $\strat(\bbf_i)$ minimizes the value in line 7, it holds that the cumulative weight $w(\sigma_{\strat'},\strat')$ is higher or equal to the cumulative weight of any run $\sigma_\strat$ under $\strat$ starting in $\bbf_i$ leading to an accepting belief state. Hence $\tac_\blf(\strat',\bbf_i)\geq\tac_\blf(\strat,\bbf_i)$ and strategy C is optimal.
}
\end{proof}

\subsection{Constructing a strategy for the NTS}\label{subsec:ntsstrat}

Let $\strat_\blf$ be the strategy for the weighted belief product $\blf$ resulting from Alg.~\ref{alg:1}. Consider the following (observation-based control and observation scheduling) strategy $\strat$ for the NTS $\nts$ with observation modes $\obsm$. For a finite sequence of observations $\sigma_\obs\in (2^\obs)^*$, we define
\begin{equation}\label{eq:strategynts}
\strat(\sigma_\obs) = \strat_\blf(\bbf),
\end{equation}
where $\bbf$ is the last state of the finite run $\sigma_\blf$ of the belief product that corresponds to $\sigma_\obs$ as described in Cor.~\ref{cor:belief}, if such run exists. 

\begin{theorem}
Let $\strat_\blf$ be the strategy for the weighted belief product $\blf$ resulting from Alg.~\ref{alg:1}. Then the strategy $\strat$ for the NTS with observation modes constructed according to Eq.~\ref{eq:strategynts} is a solution to Problem~\ref{pf:task}.
\end{theorem}

\begin{proof}
The correctness with respect to the scLTL formula $\task$ follows directly from Prop.~\ref{prop:corr1}. The optimality of $\strat$ follows from Prop.~\ref{prop:opt1} and the fact that $\tac(\strat,(s_\init,m_\init),\task)=\tac_\blf(\strat_\blf,\bbf_\init)$.
\end{proof}

\textit{Complexity.} Given an scLTL formula $\task$ the number of states of a corresponding minimal DFA $\dfa$ is in general doubly exponential in the size of the formula. \eva{However, compared to the size of the NTS, the size of the automaton typically does not play a crucial role in the overall complexity}. The product $\product$ of the NTS $\nts$ and $\dfa$ is then of size $\mathcal{O}(|S|\times |Q|)$. The belief product $\blf$ involves a subset construction over the product, hence its size is in $\mathcal{O}(2^{|S|\times |Q|})$. In order to minimize the complexity in practice, only the reachable states of both the product and the belief product are constructed. With a proper choice of a data structure storing the belief product $\blf$, Alg.~\ref{alg:1} runs in time $\mathcal{O}(|\B|\cdot \log |\B| + |\A|\cdot \mathrm{dn})$, where $\mathrm{dn}$ is the degree of non-determinism of the NTS $\nts$, \ie the maximum number of possible successors given a state and an action. \eva{Note that while the algorithms are polynomial with respect to their input, \ie the belief product, they are exponential in the size of the original NTS.} 


\section{Bounded optimal scLTL control}\label{sec:solutionbounded}

In order to solve the bounded version of Problem~\ref{pf:task}, \ie Problem~\ref{pf:taskbounded}, we proceed as follows. As in the case for the general problem, we first construct the product $\product$ of the NTS $\nts$ with a DFA $\dfa$ for the scLTL formula $\task$ and the corresponding belief product $\blf$ as proposed in Sec.~\ref{subsec:product} and \ref{subsec:belief}, respectively. To compute a strategy for the belief product from Sec.~\ref{subsec:beliefstrat}, we use an alternation of Alg.~\ref{alg:1} presented below and summarized as Alg.~\ref{alg:2}. Intuitively, as Alg.~\ref{alg:1} builds on the principles of Dijkstra's algorithm, Alg.~\ref{alg:2} follows the idea behind Bellman-Ford algorithm for solving the bounded shortest path problem in weighted graphs~\cite{algorithms}. We prove properties of the resulting strategy $\strat_\blf$ for the belief product and argue that when mapped to the original system as described in Sec.~\ref{subsec:ntsstrat}, we obtain a correct and optimal solution to Problem~\ref{pf:taskbounded}.  

\subsection{Constructing a strategy for the belief product}

\begin{algorithm}[t]
\small
\begin{algorithmic}[1]
\REQUIRE $\blf = (\B,\A, T_{\blf},\bbf_\init,\obs, F_{\blf}, w)$, bound $k\geq 1$
\ENSURE strategy $\strat$ for the belief product $\blf$ 

\STATE $W_0 := F_{\blf}$
\STATE $\forall\, \bbf \in W_0:\,\wtg(\bbf) := 0$\\ $\forall\, \bbf \in (\B \backslash W_0):\, \wtg(\bbf) := \infty$
\STATE $i:=1$

\WHILE{$i\leq k$}

\FOR{every $\bbf\in \B$}

\STATE $\abf_{\min}^\bbf:=\bot\qquad \Delta_{\min}^\bbf:=\wtg(\bbf)$
\FOR{every $\abf\in \A$ such that $\emptyset\neq T_\blf(\bbf,\abf) \subseteq W_{i-1}$}

\STATE $\Delta := \max \limits_{\bbf' \in T_\blf(\bbf,\abf)} \{w(\bbf,\abf) + \wtg(\bbf')\}$
\IF{$\Delta< \Delta_{\min}^\bbf$}
\STATE $\abf_{\min}^\bbf = \abf\qquad \Delta_{\min}^\bbf := \Delta$
\ENDIF

\ENDFOR

\ENDFOR

\STATE $W_i := W_{i-1}$

\FOR{every state $\bbf\in \B$}

\STATE $\strat(\bbf) := \abf_{\min}^\bbf$
\STATE $\wtg(\bbf) := \Delta_{\min}^\bbf$

\IF{$\wtg(\bbf)<\infty$}
\STATE $W_i:= W_i\cup \{\bbf\}$
\ENDIF

\ENDFOR

\STATE $i:=i+1$
\ENDWHILE

\IF{$\bbf_\init \in W_i$}
\RETURN $\strat$
\ELSE
\RETURN no suitable strategy exists for given bound
\ENDIF
\end{algorithmic}
\caption{Constructing a strategy for the weighted belief product and the given bound that maps to a solution of Problem~\ref{pf:taskbounded}.}\label{alg:2}
\end{algorithm}

In this section, we describe an algorithm that constructs a strategy for the weighted belief product that guarantees a visit to an accepting belief state within $k$ steps (if such a strategy exists) and minimizes the worst-case cumulative weight. 

Instead of computing the weight-to-go value $\wtg(\bbf)$ for a single well-chosen belief state $\bbf$ at a time as in Alg.~\ref{alg:1}, in Alg.~\ref{alg:2} we update the value in parallel for all states in every iteration. We show that the set $W_i$ which is the set of all belief states for which $\wtg(\bbf)\neq \infty$ after $i$-th iteration, consists of all belief states $\bbf$ that can reach an accepting belief state in at most $i$ steps and with the worst-case cumulative weight $\wtg(\bbf)$. The algorithm terminates after $k$, but at most $|\B|-1$, iterations. If the resulting set $W_i$ contains the initial belief state, the strategy consisting of the chosen actions for each belief state in $W_i$ is returned.

\begin{proposition}[Correctness]\label{prop:corr2}
Alg.~\ref{alg:2} results in a strategy $\strat$ for the weighted belief product such that every run under $\strat$ that starts in $\bbf_\init$ visits an accepting belief state in at most $k$ steps, if such a strategy exists.
\end{proposition}

\begin{proof}
The property can be proved by induction on the iteration counter $i$ and proving that starting from any state $\bbf_i\in W_i$, strategy $\strat$ guarantees a visit to an accepting belief state in at most $i$ steps.
\end{proof}

\begin{proposition}[Optimality]\label{prop:opt2}
Let $\strat$ be the strategy resulting from Alg.~\ref{alg:2}. Then among all strategies that guarantee a visit to an accepting belief state in at most $k$ steps, $\strat$ minimizes the value in Eq.~\ref{eq:opt}.
\end{proposition}

\begin{proof}
\eva{Let $C_i$ and $\wtg_i$ denote the strategy and the weight-to-go computed by Alg.~\ref{alg:2} before start of the $(i+1)$-th iteration. We show by induction that after every iteration $i\geq 1$, it holds that $\tac_\blf(\strat_i,\bbf_i)=\wtg_i(\bbf_i)$ for all states $\bbf_i\in W_i$, and that strategy $\strat_i$ minimizes the value $\tac_\blf(\cdot,\bbf_i)$ among all strategies that guarantee visit to an accepting state in at most $i$ steps.}

\eva{
Assume that for all $j< i$ it holds $\tac_\blf(\strat_j,\bbf_j)=\wtg_j(\bbf_j)$ for all states $\bbf_j\in W_j$, and that strategy $\strat_j$ minimizes the value $\tac_\blf(\cdot,\bbf_j)$ among all strategies that guarantee visit to an accepting state in at most $j$ steps. Trivially, $\strat_0$ minimizes the value for all accepting belief states $\bbf\in F_\blf$ as $\tac_\blf(\strat_0,\bbf)=0=\wtg_0(\bbf)$. We show that $\strat_i$ then minimizes the value $\tac_\blf(\cdot,\bbf_i)$ over all strategies that guarantee visit to an accepting state in at most $i$ steps and $\tac_\blf(\strat_i,\bbf_i) = \wtg_i(\bbf_i)$.}

\eva{
Assume by contradiction that there exists a strategy $\strat'$ for $\blf$  such that $\tac_\blf(\strat',\bbf_i)<\tac_\blf(\strat_i,\bbf_i)$ and guarantee visit to an accepting state for all $\bbf_i \in W_i$ in at most $i$ steps. Recall that $W_{i-1}$ contains exactly those states that guarantee visit to an accepting state in at most $i-1$ steps (see proof of Prop.~\ref{prop:corr2}). Let $\bbf_i$ be an arbitrary belief state from $W_i$. Since $C'$ is winning in at most $i$ steps, it holds $T_\blf(\bbf_i,C'(\bbf_i)) \subseteq W_{i-1}$. From the induction hypothesis and the fact the action $C_i(\bbf_i)$ minimizes the value in line 7 it follows that $\tac_\blf(\strat',\bbf_i)\geq\wtg_i(\bbf_i)=\tac_\blf(\strat_i,\bbf_i)$ implying strategy $C_i$ is optimal strategy among all strategies that guarantee visit to an accepting state in at most $i$ steps.}
\end{proof}

\begin{theorem}
Let $\strat_\blf$ be the strategy for the weighted belief product $\blf$ resulting from Alg.~\ref{alg:2}. Then the strategy $\strat$ for the NTS with observation modes constructed according to Eq.~\ref{eq:strategynts} is a solution to Problem~\ref{pf:taskbounded}.
\end{theorem}

\begin{proof}
The correctness with respect to the scLTL formula $\task$ follows directly from Prop.~\ref{prop:corr2} and the optimality of $\strat$ follows from Prop.~\ref{prop:opt2} and the fact that $\tac(\strat,(s_\init,m_\init),\task)=\tac_\blf(\strat_\blf,\bbf_\init)$.
\end{proof}

\textit{Complexity.} The size of a minimal DFA for $\task$, the product and the belief product are discussed in Sec.~\ref{subsec:ntsstrat}. Similarly as for Alg.~\ref{alg:1}, with a proper choice of a data structure storing the belief product $\blf$, Alg.~\ref{alg:2} runs in time $\mathcal{O}(k\cdot |\B|\cdot |\A|\cdot \mathrm{dn})$, where $\mathrm{dn}$ is the degree of non-determinism of $\nts$, \ie the maximum number of possible successors given a state and an action. Here, the value $|\B|\cdot |\A|\cdot \mathrm{dn}$ serves as the upper bound on the number of all edges in the belief product. Note that  Alg.~\ref{alg:2} can be terminated prematurely if the current iteration $i$ of the while loop in line 4 did not imply any change in the function $\wtg$ or if $i\geq |\B|-1$ as every strategy for the belief product that guarantees a visit to an accepting belief state must do so in at most $|\B|-1$ steps due to non-determinism. 

\begin{remark}
Note that Alg.~\ref{alg:2} can be used not only to solve the bounded Problem~\ref{pf:taskbounded}, but also the general Problem~\ref{pf:task} by considering $k=|\B|-1$. However, this solution to Problem~\ref{pf:task} has higher computational complexity in practice than the one presented in Sec.~\ref{sec:solution} using Alg.~\ref{alg:1}.
\end{remark}


\section{Case Study}\label{sec:cs}

\newcommand{\north}{\mathtt{N}}
\newcommand{\south}{\mathtt{S}}
\newcommand{\east}{\mathtt{E}}
\newcommand{\west}{\mathtt{W}}
\newcommand{\northA}{\mathsf{N}}
\newcommand{\southA}{\mathsf{S}}
\newcommand{\eastA}{\mathsf{E}}
\newcommand{\westA}{\mathsf{W}}
\newcommand{\northwest}{\mathtt{NW}}
\newcommand{\northeast}{\mathtt{NE}}
\newcommand{\southwest}{\mathtt{SW}}
\newcommand{\southeast}{\mathtt{SE}}
\newcommand{\detected}{\mathtt{det}}
\newcommand{\dang}{\mathtt{dang}}
\newcommand{\target}{\mathtt{target}}


\begin{figure}[t]
\begin{center}
\begin{tabular}{c c c}
\scalebox{0.15}{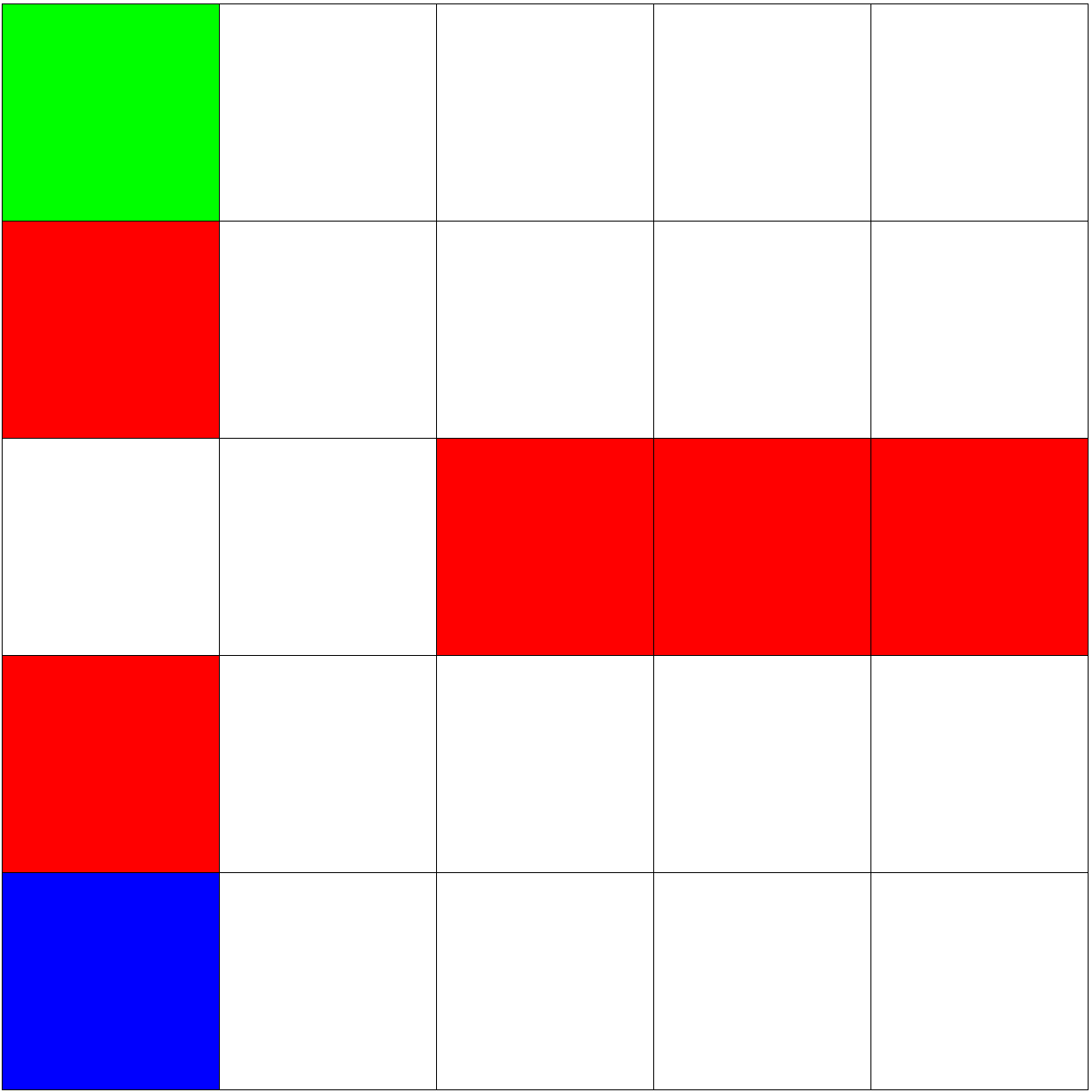}  & \scalebox{0.15}{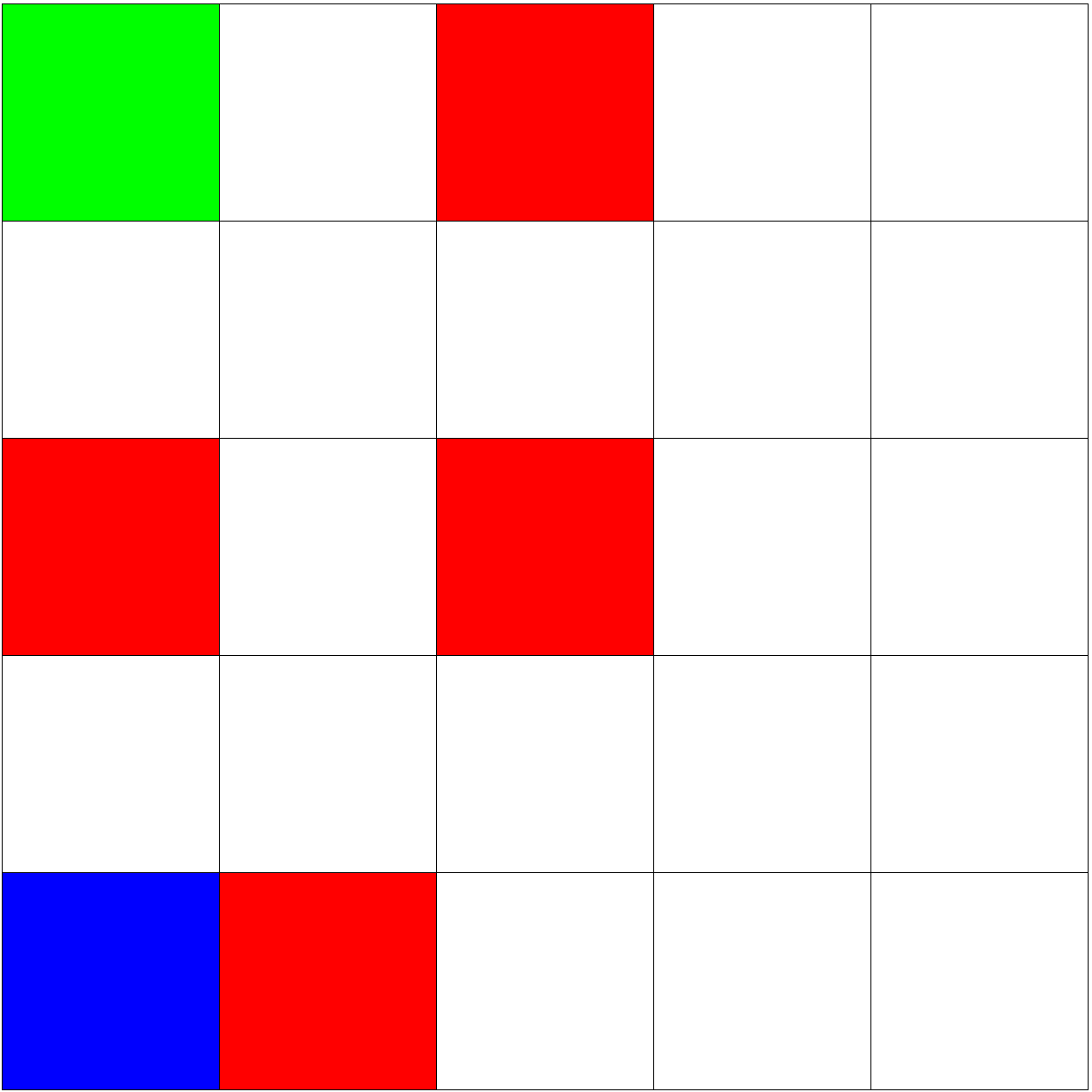} & \scalebox{0.15}{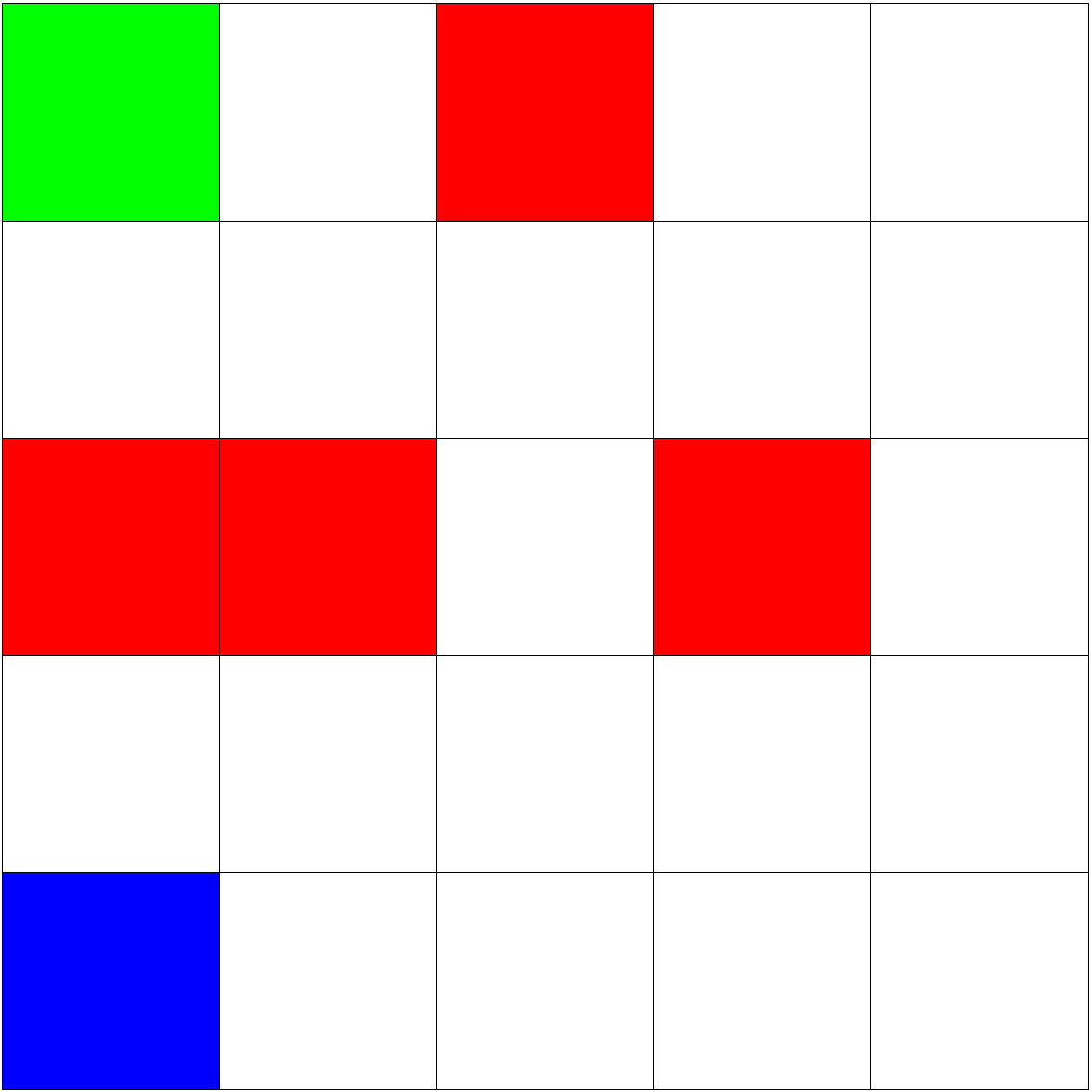}\\
1 & 2 & 3
\end{tabular}
\end{center}
\caption{The environment of a mobile robot partitioned into a grid of $5\times 5$ equally sized regions. The three grids correspond to three possible placements of dangerous regions, shown in red. The locations of the starting region, in green, and the target region, in blue, are known and hence their placement is the same in all three grids.}\label{fig:grids}
\end{figure}

We implemented the algorithms from Sec.~\ref{sec:solution} and~\ref{sec:solutionbounded} in C++. In this section, we demonstrate their use on a case study motivated by examples in~\cite{pomdpexamples,pomdpcasestudy}. All executions were performed on Mac OS X 10.10.3 with 2.6 GHz Intel Core i5 processor and 8 GB 1600 MHz DDR3 memory. 

Consider a mobile robot moving in an environment partitioned into a grid of $5\times 5$ equally sized regions. The grid contains a starting, a target and possibly multiple dangerous regions, where the robot is detected and captured. The robot knows the locations of the starting and the target regions but it does not know the exact locations of dangerous regions. Nevertheless, the robot knows that the grid takes one of the three forms depicted in Fig.~\ref{fig:grids}. The robot moves deterministically in (up to) four directions corresponding to the movement in the four compass directions. To learn the presence of dangerous regions in it's immediate surroundings, the robot can deploy one of the two sensors in Fig.~\ref{fig:sensors}. The first sensor partitions the neighboring area into quadrants and reports the set of all quadrants that contain at least one dangerous region. The second sensor reports the exact regions in robot's immediate surroundings that are dangerous. In every step, the robot can decide which sensor to activate, if any. The costs of deployment of the two sensors is 1 and 2, respectively. The cost can be interpreted as the amount of resources needed for the use of each sensor. Alternatively, it may model the amount of information received by the enemy in dangerous regions. The goal of the robot is to reach the target region from the starting region without being detected, while minimizing the cost. 

\begin{figure}[t]
\begin{center}
\begin{tabular}{c c c}
\scalebox{0.23}{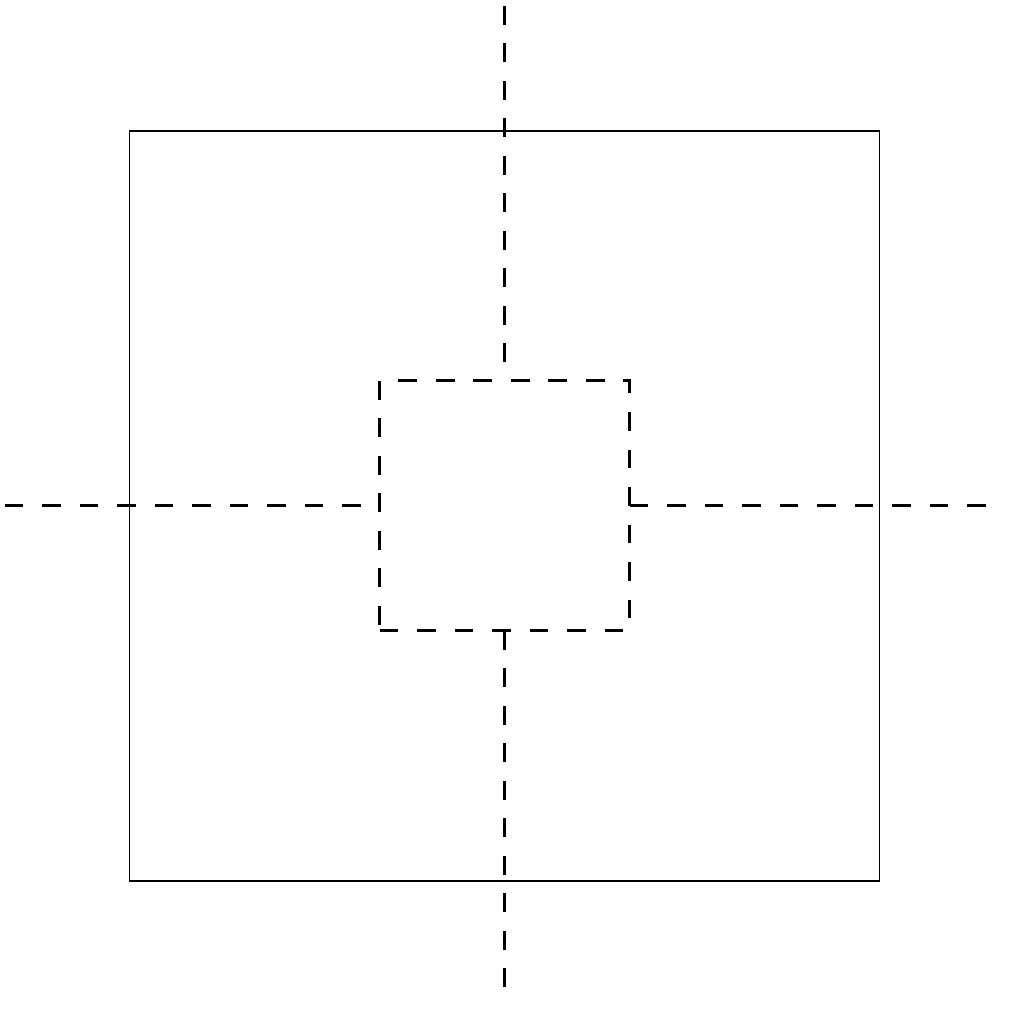}  & \scalebox{0.23}{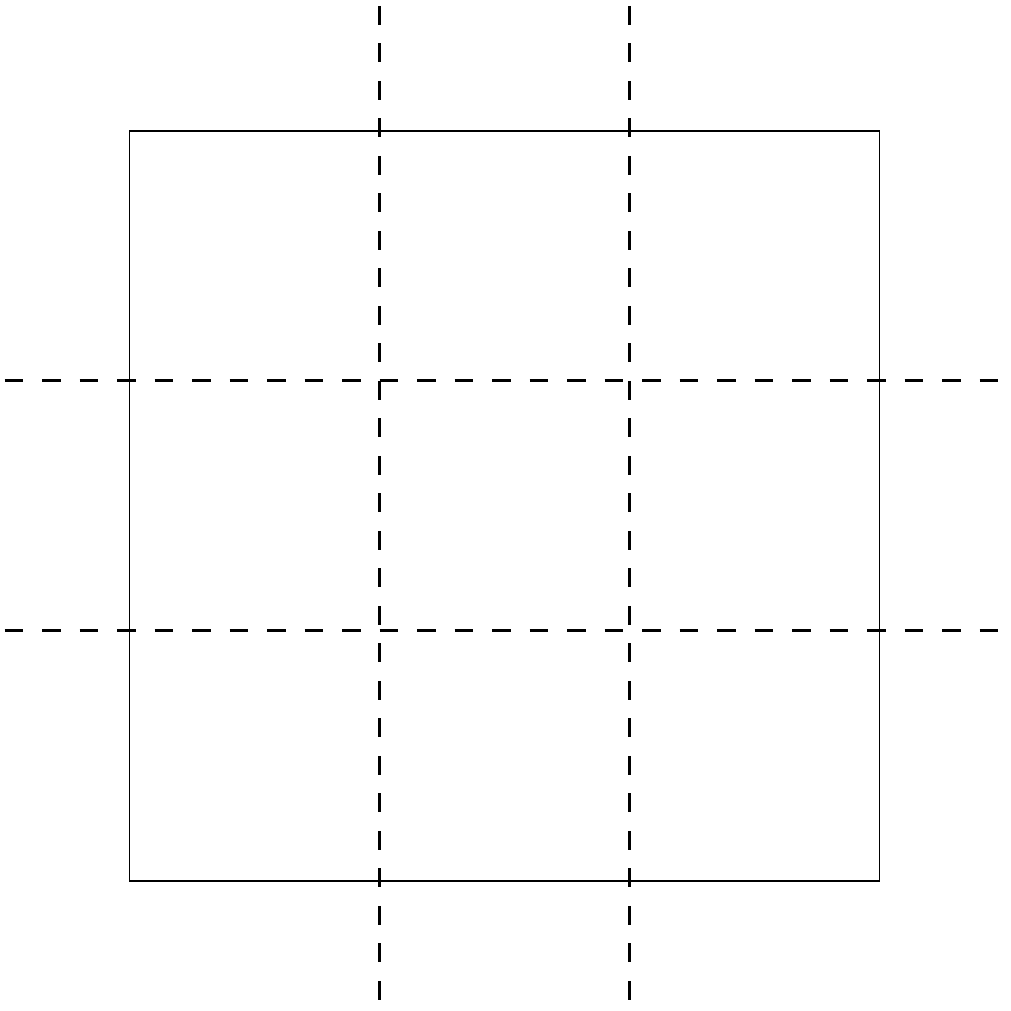} & \scalebox{0.23}{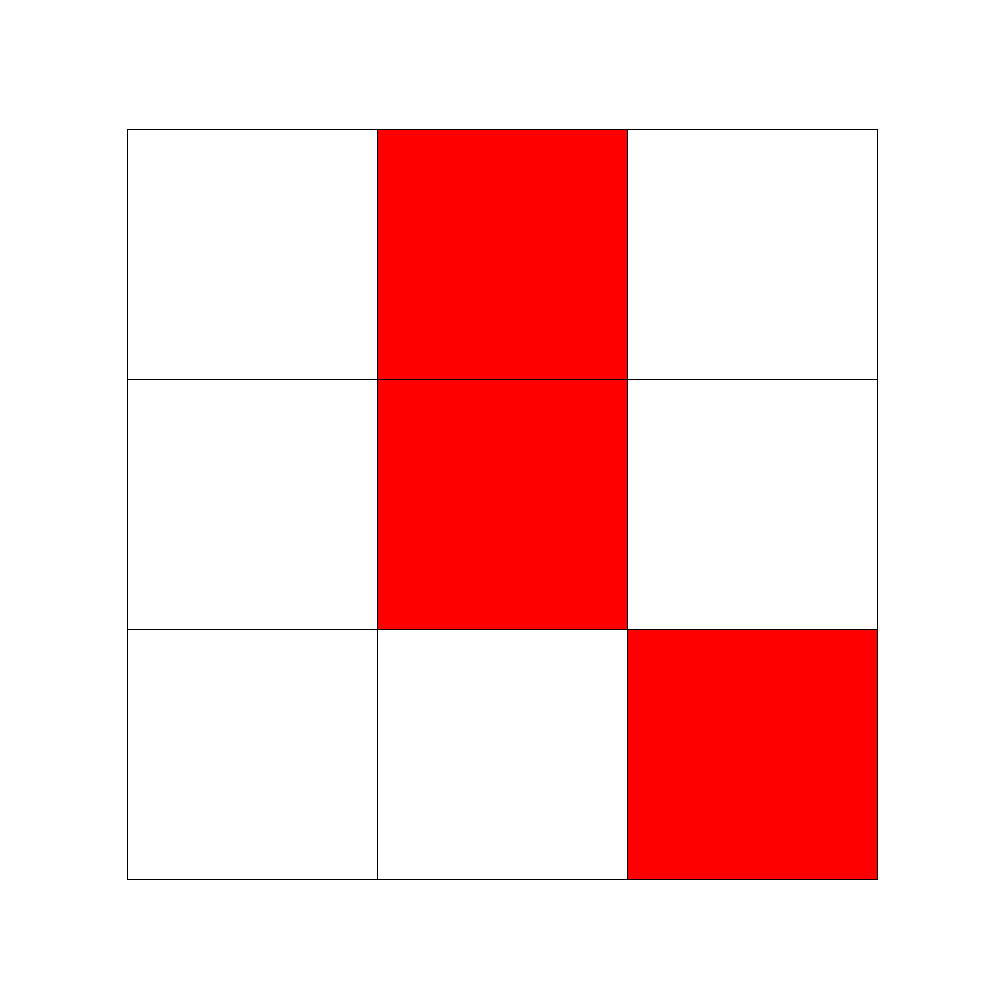}\\
(a) & (b) & (c)
\end{tabular}
\end{center}
\caption{Two sensors that provide information about the presence of dangerous regions in robot's immediate surroundings. We also show the names of the corresponding observations learned by the robot. For the first sensor in (a), the surrounding area is divided into quadrants and the sensor reports all quadrants containing a dangerous region. For the second sensor in (b), the exact set of dangerous regions is reported. For example, let (c) show the immediate surroundings of the robot with it's current position in the middle and dangerous regions in red. The first sensor reports the set of observations $\{\northwest,\northeast,\southeast,\det\}$ and the second sensor reports $\{\north,\southeast,\det\}$.}\label{fig:sensors}
\end{figure}

\eva{The NTS with observation modes} that models the above system has 76 states $S=\{s_\init,s_{ijk}\mid 1\leq i\leq 3, 1\leq j,k\leq 5\}$ and 5 actions $A=\{a,\northA,\southA,\eastA,\westA\}$. States $s_{ijk}$ correspond to the regions in the three grids, where $1\leq i\leq 3$ is the grid identifier and $1\leq j,k\leq 5$ determine the row and column coordinate, respectively. For example, $s_{111}$ is the top left corner of the first grid. The initial state $s_\init$ has only one transition $T(s_\init,a)=\{s_{111},s_{211},s_{311}\}$ that corresponds to the enemy choosing one of the three grids in Fig.~\ref{fig:grids}. The transitions of all $s_{ijk}$ are deterministic and correspond to moving in compass directions $\northA,\southA,\eastA,\westA$. The set $\ap=\{\dang,\target\}$ and the labeling function is such that $L(s_\init)=\emptyset$ and $L(s_{ijk})$ indicates the target and dangerous regions as in Fig.~\ref{fig:grids}. The set of observations is $\obs = \{\north,\south,\west,\east,\northwest,\northeast,\southwest,\southeast,\det\}$. The NTS has 3 observation modes corresponding to not activating any sensor, activating the first sensor and activating the second sensor, respectively. The respective observation functions $\obsf_1,\obsf_2,\obsf_3$ are defined in Fig.~\ref{fig:sensors} and $\cost_1=0,\cost_2=1,\cost_3=2$. Note that in every step of an execution of the system, we know the robot's position in the grid precisely, only the identifier of the grid is unknown.

The scLTL formula specifying the robot's mission is $(\neg \dang)\U \target$ and the corresponding minimal DFA $\dfa$ has 3 states. The product $\product$ of $\nts$ and $\dfa$ has 208 states and 667 (possibly non-deterministic) transitions, and was constructed in less than 0.1 seconds. The weighted belief product $\blf$ has 375 states and 2634 transitions, and was constructed in 1.5 seconds. 

Using Alg.~\ref{alg:1}, we computed an optimal solution to Problem~\ref{pf:task} for this system in 7 seconds. The corresponding strategy for the robot is as follows. In the starting region, use the first sensor. If the reported observations are $\southeast$ and $\southwest$, then the robot is in grid 1 from Fig.~\ref{fig:grids}. If the set of reported observations is empty, the robot is moving either in grid 2 or grid 3. In the former case, do not deploy any sensors anymore and move in directions $\eastA,\southA,\southA,\southA,\southA,\westA$ to reach the target region. In the latter case, do not use any sensor anymore and move in directions $\southA,\eastA,\eastA,\eastA,\eastA,\southA,\southA,\southA,\westA,\westA,\northA,\westA,\westA,\southA$. The worst-case cost of the strategy is 1 and the maximum number of transitions performed by the robot to reach the target region is 14, \ie in the NTS $\nts$ it is 15 including the first step for choosing the grid.


Next, we used Alg.~\ref{alg:2} to solve the bounded version of the problem for bounds $k<15$, where the mission must be satisfied faster than using the strategy above. For all choices of $k$ below, the algorithm terminated in less than 3 seconds. For $k\leq 8$, there does not exist a suitable strategy. For $k=13$ and $k=14$, the optimal strategy has the same structure as the one resulting from Alg.~\ref{alg:1} with the following exception. If the robot learns that it moves either in grid 2 or 3, the sequence of directions is $\southA,\eastA,\eastA,\eastA,\eastA,\southA,\southA,\westA,\westA,\westA,\westA,\southA$. The maximum number of steps needed to reach the target region is 13 and the worst-case cost of the strategy is 1.

Finally, for $9\leq k\leq 12$ there exists a solution and the corresponding optimal strategy for the robot is as follows. From the starting region, move in directions $\eastA$ and then $\southA$ without deploying any sensor. Then move in direction $\eastA$ and activate the second sensor. If the reported observations are $\south$ and $\southeast$ then the robot is moving in grid 1. In such a case, do not use any sensors anymore and move in directions $\westA,\southA,\southA,\southA,\westA$ to reach the target region. Similarly, if the observations are $\south$ an $\north$ then the robot is in grid 2, do not use any sensors anymore and move in directions $\westA,\southA,\southA,\westA,\southA$. Finally, if the observations are $\southwest, \southeast$ and $\north$ then the robot is in grid 3, do not use any sensors and move in directions $\southA,\southA,\southA,\westA,\westA$. While the maximum number of steps needed to reach the target region is 9, the worst-case cost of the strategy is 2.


%



\section{Conclusion and Future Work}\label{sec:conclusion}

We consider non-deterministic transition systems with multiple observation modes with fixed non-negative costs. We present correct and optimal algorithms to solve two optimal temporal control problems. The first aims to construct a control and observation mode switching strategy that guarantees satisfaction of a finite-time temporal property given as a formula of scLTL and minimizes the worst-case cost accumulated until the point of satisfaction. Second, we consider the bounded version of the problem with a bound on the time of satisfaction. Both algorithms are demonstrated on a case study motivated by robotic application.

In our future work, we aim to use the results presented in this work as a basis for solving more intriguing problems of observation scheduling under temporal constraints. The extensions include infinite-time temporal properties, infinite-time optimization objectives and probabilistic models.

\bibliographystyle{abbrv}
\bibliography{references}{}

\begin{thebibliography}{10}

\bibitem{baierbook}
C.~Baier and J.-P. Katoen.
\newblock {\em {Principles of Model Checking}}.
\newblock The MIT Press, 2008.

\bibitem{bertrand11}
N.~Bertrand and B.~Genest.
\newblock {Minimal Disclosure in Partially Observable Markov Decision Processes
  }.
\newblock In {\em Proc. of {FSTTCS}}, volume~13, pages 411--422, 2011.

\bibitem{martinaaai15}
K.~Chatterjee, M.~Chmelik, R.~Gupta, and A.~Kanodia.
\newblock {Optimal Cost Almost-Sure Reachability in POMDPs}.
\newblock In {\em {Proc. of AAAI}}, pages 3496--3502, 2015.

\bibitem{pomdpexamples}
K.~Chatterjee, M.~Chmelik, R.~Gupta, and A.~Kanodia.
\newblock {Qualitative analysis of POMDPs with temporal logic specifications
  for robotics applications}.
\newblock In {\em Proc. of {ICRA}}, pages 325--330, 2015.

\bibitem{pogamessurvey}
K.~Chatterjee, L.~Doyen, and T.~A. Henzinger.
\newblock A survey of partial-observation stochastic parity games.
\newblock {\em {Formal Methods in System Design}}, 43(2):268--284, 2013.

\bibitem{krishminattention}
K.~Chatterjee and R.~Majumdar.
\newblock {Minimum Attention Controller Synthesis for Omega-Regular
  Objectives}.
\newblock In {\em Proc. of {FORMATS}}, volume 6919 of {\em Lecture Notes in
  Computer Science}, pages 145--159, 2011.

\bibitem{krishbudget}
K.~Chatterjee, R.~Majumdar, and T.~Henzinger.
\newblock {Controller Synthesis with Budget Constraints}.
\newblock In {\em Proc. of {HSCC}}, volume 4981 of {\em Lecture Notes in
  Computer Science}, pages 72--86, 2008.

\bibitem{algorithms}
T.~H. Cormen, C.~Stein, R.~L. Rivest, and C.~E. Leiserson.
\newblock {\em {Introduction to Algorithms}}.
\newblock {McGraw-Hill Higher Education}, 2nd edition, 2001.

\bibitem{hybridsensors}
Z.~Feng, K.~Teo, and V.~Rehbock.
\newblock Hybrid method for a general optimal sensor scheduling problem in
  discrete time.
\newblock {\em Automatica}, 44(5):1295 -- 1303, 2008.

\bibitem{austin}
A.~Jones, M.~Schwager, and C.~Belta.
\newblock A receding horizon algorithm for informative path planning with
  temporal logic constraints.
\newblock In {\em Proc. of {ICRA}}, pages 5019--5024, 2013.

\bibitem{kaelbling98}
L.~P. Kaelbling, M.~L. Littman, and A.~R. Cassandra.
\newblock Planning and acting in partially observable stochastic domains.
\newblock {\em Artificial Intelligence}, 101(1–2):99 -- 134, 1998.

\bibitem{scltl}
O.~Kupferman and M.~Y.~Vardi.
\newblock {Model Checking of Safety Properties}.
\newblock {\em Formal Methods in System Design}, 19(3):291--314, 2001.

\bibitem{madani03}
O.~Madani, S.~Hanks, and A.~Condon.
\newblock On the undecidability of probabilistic planning and related
  stochastic optimization problems.
\newblock {\em Artificial Intelligence}, 147(1–2):5 -- 34, 2003.

\bibitem{mrg}
P.~Ondruska, C.~Gurau, L.~Marchegiani, C.~H. Tong, and I.~Posner.
\newblock {Scheduled Perception for Energy-Efficient Path Following}.
\newblock In {\em Proc. of {ICRA}}, 2015.

\bibitem{pineau03}
J.~Pineau, G.~Gordon, and S.~Thrun.
\newblock {Point-based Value Iteration: An Anytime Algorithm for POMDPs}.
\newblock In {\em {Proc. of IJCAI}}, pages 1025--1030, 2003.

\bibitem{ltl}
A.~Pnueli.
\newblock The temporal logic of programs.
\newblock In {\em 18th Annual Symposium on Foundations of Computer Science},
  pages 46--57, Oct 1977.

\bibitem{sipser}
M.~Sipser.
\newblock {\em Introduction to the theory of computation}.
\newblock {PWS} Publishing Company, 1997.

\bibitem{pomdpcasestudy}
M.~Svorenova, M.~Chmelik, K.~Leahy, H.~F. Eniser, K.~Chatterjee, I.~Cerna, and
  C.~Belta.
\newblock {Temporal logic motion planning using POMDPs with parity objectives:
  case study paper}.
\newblock In {\em Proc. of {HSCC}}, pages 233--238, 2015.

\bibitem{alphanscltl}
A.~Ulusoy, T.~Wongpiromsarn, and C.~Belta.
\newblock {Incremental Controller Synthesis in Probabilistic Environments with
  Temporal Logic Constraints}.
\newblock {\em Int. Journal of Robotics Research}, 33(8):1130--1144, 2014.

\bibitem{linearsensors}
M.~Vitus, W.~Zhang, A.~Abate, J.~Hu, and C.~Tomlin.
\newblock On efficient sensor scheduling for linear dynamical systems.
\newblock In {\em Proc. of {ACC}}, pages 4833--4838, 2010.

\end{thebibliography}

\end{document}